\documentclass[12pt,reqno]{article}

\usepackage[usenames]{color}
\usepackage{amssymb}
\usepackage{graphicx}
\usepackage{amscd}

\usepackage{graphics,amsmath,amssymb,relsize,mathalfa}
\usepackage{bm}
\usepackage{amsthm}
\usepackage{amsfonts}
\usepackage{latexsym}
\usepackage{epsf}

\newtheorem{theorem}{Theorem}[section]
\newtheorem{corollary}[theorem]{Corollary}
\newtheorem{lemma}[theorem]{Lemma}
\newtheorem{proposition}[theorem]{Proposition}

\newtheorem{defin}[theorem]{Definition}
\newenvironment{definition}{\begin{defin}\normalfont\quad}{\end{defin}}
\newtheorem{examp}[theorem]{Example}

\newtheorem{rema}[theorem]{Remark}
\newtheorem{prob}[theorem]{Problem}

\numberwithin{equation}{section}

\newcommand{\bt}{\begin{thm}}
\newcommand{\et}{\end{thm}}
\newcommand{\bp}{\begin{proof}}
\newcommand{\ep}{\end{proof}}
\newcommand{\bprop}{\begin{prop}}
\newcommand{\eprop}{\end{prop}}
\newcommand{\bl}{\begin{lemma}}
\newcommand{\el}{\end{lemma}}
\newcommand{\bc}{\begin{corollary}}
\newcommand{\ec}{\end{corollary}}
\newcommand{\Z}{\mathbb{Z}}

\newcommand{\be}{\begin{enumerate}}
\newcommand{\ee}{\end{enumerate}}

\title{MMH$^*$ with arbitrary modulus is always almost-universal}
\author{Khodakhast Bibak \thanks{Department of Computer Science, University of Victoria, Victoria, BC, Canada V8W 3P6. Email: {\tt
\{kbibak,bmkapron,srinivas\}@uvic.ca}} \and Bruce M. Kapron \footnotemark[1] \and Venkatesh Srinivasan \footnotemark[1] \thanks{Centre for Quantum Technologies, National University of Singapore, Singapore 117543.}}

\begin{document}

\maketitle

\begin{abstract}
Universal hash functions, discovered by Carter and Wegman in 1979, are of great importance in computer science with many applications. MMH$^*$ is a well-known $\triangle$-universal hash function family, based on the evaluation of a dot product modulo a prime. In this paper, we introduce a generalization of MMH$^*$, that we call GMMH$^*$, using the same construction as MMH$^*$ but with an arbitrary integer modulus $n>1$, and show that GMMH$^*$ is $\frac{1}{p}$-almost-$\triangle$-universal, where $p$ is the smallest prime divisor of $n$. This bound is tight.
\end{abstract}


\section{MMH$^*$}\label{Sec 1}

Universal hashing, introduced by Carter and Wegman \cite{CW}, is of great importance in computer science with many applications. Cryptography, information security, complexity theory, randomized algorithms, and data structures are just a few areas that universal hash functions and their variants have been used as a fundamental tool. In \cite{HK}, definitions of various kinds of universal hash functions gathered from the literature are presented; we mention some of them here.

\begin{definition} \label{def:Uni hash}
Let $H$ be a family of functions from a domain $D$ to a range $R$. Let $\varepsilon$ be a constant such that $\frac{1}{|R|} \leq \varepsilon < 1$. The probabilities below are taken over the random choice of hash function $h$ from the set $H$.

\noindent\textit{(i)} The family $H$ is a {\it universal family of hash functions} if for any two distinct $x,y\in D$, we have $\text{Pr}_{h \leftarrow H}[h(x) = h(y)] \leq \frac{1}{|R|}$. Also, $H$ is an {\it $\varepsilon$-almost-universal} ($\varepsilon$-AU) {\it family of hash functions} if for any two distinct $x,y\in D$, we have $\text{Pr}_{h \leftarrow H}[h(x) = h(y)] \leq \varepsilon$.

\noindent\textit{(ii)} Suppose $R$ is an Abelian group. The family $H$ is a $\triangle$-{\it universal family of hash functions} if for any two distinct $x,y\in D$, and all $b\in R$, we have $\text{Pr}_{h \leftarrow H}[h(x) - h(y) = b] = \frac{1}{|R|}$, where ` $-$ ' denotes the group subtraction operation. Also, $H$ is an {\it $\varepsilon$-almost-$\triangle$-universal} ($\varepsilon$-A$\triangle$U) {\it family of hash functions} if for any two distinct $x,y\in D$, and all $b\in R$, we have $\text{Pr}_{h \leftarrow H}[h(x) - h(y) = b] \leq \varepsilon$.
\end{definition}

It is worth mentioning that $\varepsilon$-A$\triangle$U families have also important applications in computer science, in particular, in cryptography. For example, these families can be used in message authentication. Informally, it is possible to design a message authentication scheme using $\varepsilon$-A$\triangle$U families such that two parties can exchange signed messages over an unreliable channel and the probability that an adversary can forge a valid signed message to be sent across the channel is at most $\varepsilon$ (\cite{HK}).

The following family, named MMH$^*$ by Halevi and Krawczyk \cite{HK} in 1997, is a well-known $\triangle$-universal hash function family.

\begin{definition}\label{def:MMH$^*$}
Let $p$ be a prime and $k$ be a positive integer. The family MMH$^*$ is defined as follows:
\begin{align}\label{MMH*}
\text{MMH}^*:=\lbrace g_{\mathbf{x}} \; : \; \mathbb{Z}_p^k \rightarrow \mathbb{Z}_p \; | \; \mathbf{x}\in \mathbb{Z}_p^k \rbrace,
\end{align}
where
\begin{align}\label{MMH* for 2}
g_{\mathbf{x}}(\mathbf{m}) := \mathbf{m} \cdot \mathbf{x} \pmod{p} = \sum_{i=1}^k m_ix_i \pmod{p},
\end{align}
for any $\mathbf{x}=\langle x_1, \ldots, x_k \rangle \in \mathbb{Z}_p^k$, and any $\mathbf{m}=\langle m_1, \ldots, m_k \rangle \in \mathbb{Z}_p^k$.
\end{definition}

It appears that Gilbert, MacWilliams, and Sloane \cite{GMS} first discovered MMH$^*$ (but in the finite geometry setting). However, many resources attribute MMH$^*$ to Carter and Wegman \cite{CW}. Halevi and Krawczyk \cite{HK} proved that MMH$^*$ is a $\triangle$-universal family of hash functions. We also remark that, recently, Leiserson et al. \cite{LSS} rediscovered MMH$^*$ (called it ``DOTMIX compression function family") and using the same method as Halevi and Krawczyk, proved that DOTMIX is $\triangle$-universal. They then apply this result to the problem of deterministic parallel random-number generation for dynamic multithreading platforms in parallel computing.

\begin{theorem} \label{thm:MMH* UNI}
The family \textnormal{MMH}$^*$ is a $\triangle$-universal family of hash functions.
\end{theorem}

\section{GMMH$^*$}\label{Sec 2}

Given that, in the definition of MMH$^*$, the modulus is a prime, it is natural to ask what happens if the modulus is an arbitrary integer $n>1$. Is the resulting family still $\triangle$-universal? If not, what can we say about $\varepsilon$-almost-universality or $\varepsilon$-almost-$\triangle$-universality of this new family? This is an interesting and natural problem, and while it has a simple solution (see, Theorem~\ref{thm:GMMH* UNI} below), to the best of our knowledge there are no results regarding this problem in the literature.

First, we define a generalization of MMH$^*$, namely, GMMH$^*$, with the same construction as MMH$^*$ except that we use an arbitrary integer $n>1$ instead of prime $p$.

\begin{definition}\label{def:GMMH$^*$}
Let $n$ and $k$ be positive integers $(n>1)$. The family GMMH$^*$ is defined as follows:
\begin{align}\label{GMMH*}
\text{GMMH}^*:=\lbrace h_{\mathbf{x}} \; : \; \mathbb{Z}_n^k \rightarrow \mathbb{Z}_n \; | \; \mathbf{x}\in \mathbb{Z}_n^k \rbrace,
\end{align}
where
\begin{align}\label{GMMH* for 2}
h_{\mathbf{x}}(\mathbf{m}) := \mathbf{m} \cdot \mathbf{x} \pmod{n} = \sum_{i=1}^k m_ix_i \pmod{n},
\end{align}
for any $\mathbf{x}=\langle x_1, \ldots, x_k \rangle \in \mathbb{Z}_n^k$, and any $\mathbf{m}=\langle m_1, \ldots, m_k \rangle \in \mathbb{Z}_n^k$.
\end{definition}

MMH$^*$ has found important applications, however, in applications that, for some reasons, we have to work in the ring $\Z_n$, the family GMMH$^*$ may be used.

The following result, proved by D. N. Lehmer \cite{LEH2}, is the main ingredient in the proof of 
Theorem~\ref{thm:GMMH* UNI}.

\begin{proposition}\label{Prop: lin cong}
Let $a_1,\ldots,a_k,b,n\in \Z$, $n\geq 1$. The linear congruence $a_1x_1+\cdots +a_kx_k\equiv b \pmod{n}$ has a solution $\langle x_1,\ldots,x_k \rangle \in \Z_{n}^k$ if and only if $\ell \mid b$, where
$\ell=\gcd(a_1, \ldots, a_k, n)$. Furthermore, if this condition is satisfied, then there are $\ell n^{k-1}$ solutions.
\end{proposition}

Now, we are ready to state and prove the following result about $\varepsilon$-almost-$\triangle$-universality of GMMH$^*$.

\begin{theorem} \label{thm:GMMH* UNI}
Let $n$ and $k$ be positive integers $(n>1)$. The family \textnormal{GMMH}$^*$ is $\frac{1}{p}$-\textnormal{A}$\triangle$\textnormal{U}, where $p$ is the smallest prime divisor of $n$. This bound is tight.
\end{theorem}

\begin{proof}
Suppose that $n$ has the prime factorization $n=p_1^{r_1}\ldots p_s^{r_s}$, where $p_1 < \cdots < p_s$ are primes and $r_1,\ldots,r_s$ are positive integers. Let $\mathbf{m}=\langle m_1, \ldots, m_k \rangle \in \mathbb{Z}_n^k$ and $\mathbf{m}'= \langle m'_1, \ldots, m'_k \rangle \in \mathbb{Z}_n^k$ be any two distinct messages. Put $\mathbf{a}=\langle a_1,\dots,a_k \rangle = \mathbf{m}-\mathbf{m}'$. For every $b\in \mathbb{Z}_n$ we have
\begin{align*}
h_{\mathbf{x}}(\mathbf{m})-h_{\mathbf{x}}(\mathbf{m'})=b \Longleftrightarrow \sum_{i=1}^k m_ix_i - \sum_{i=1}^k m'_ix_i \equiv b \pmod{n} \Longleftrightarrow \sum_{i=1}^k a_ix_i \equiv b \pmod{n}.
\end{align*}
Note that since $\langle x_1,\ldots,x_k \rangle \in \Z_{n}^k$, we have $n^k$ ordered $k$-tuples 
$\langle x_1, \ldots, x_k \rangle$. Also, since $\mathbf{m}\not=\mathbf{m}'$, there exists some $i_0$ such that $a_{i_0} \not= 0$. Now, we need to find the maximum number of solutions of the above linear congruence over all choices of $\mathbf{a}=\langle a_1,\dots,a_k \rangle \in \mathbb{Z}_n^k \setminus \lbrace \mathbf{0} \rbrace$ and $b\in \mathbb{Z}_n$. By Proposition~\ref{Prop: lin cong}, if $\ell=\gcd(a_1, \ldots, a_k, n) \nmid b$ then the linear congruence $a_1x_1+\cdots +a_kx_k\equiv b \pmod{n}$ has no solution, and if $\ell=\gcd(a_1, \ldots, a_k, n) \mid b$ then the linear congruence has $\ell n^{k-1}$ solutions. Thus, we need to find the maximum of $\ell=\gcd(a_1, \ldots, a_k, n)$ over all choices of 
$\mathbf{a}=\langle a_1,\dots,a_k \rangle \in \mathbb{Z}_n^k \setminus \lbrace \mathbf{0} \rbrace$. Clearly, 
$$
\max_{\mathbf{a}=\langle a_1,\dots,a_k \rangle \in \mathbb{Z}_n^k \setminus \lbrace \mathbf{0} \rbrace}\gcd(a_1, \ldots, a_k, n)
$$
is achieved when $a_{i_0}=p_1^{r_1-1}p_2^{r_2}\ldots p_s^{r_s}=\frac{n}{p_1}$, and $a_i=0$ ($i \not= i_0$). So, we get
$$
\max_{\mathbf{a}=\langle a_1,\dots,a_k \rangle \in \mathbb{Z}_n^k \setminus \lbrace \mathbf{0} \rbrace} 
\gcd(a_1, \ldots, a_k, n)=p_1^{r_1-1}p_2^{r_2}\ldots p_s^{r_s}=\frac{n}{p_1}.
$$
Therefore, for any two distinct messages $\mathbf{m}, \mathbf{m}' \in \mathbb{Z}_n^k$, and all $b\in \mathbb{Z}_n$, we have
\begin{align*}
\text{Pr}_{h_{\mathbf{x}} \leftarrow \text{GMMH}^*}[h_{\mathbf{x}}(\mathbf{m})-h_{\mathbf{x}}(\mathbf{m'})=b] \leq \max_{\mathbf{a}=\langle a_1,\dots,a_k \rangle \in \mathbb{Z}_n^k \setminus \lbrace \mathbf{0} \rbrace} \frac{n^{k-1}\gcd(a_1, \ldots, a_k, n)}{n^k}= \frac{1}{p_1}.
\end{align*}
This means that GMMH$^*$ is $\frac{1}{p_1}$-A$\triangle$U. Clearly, this bound is tight; take, for example, $a_1=\frac{n}{p_1}$ and $a_2=\cdots=a_k=0$.
\end{proof}

\begin{corollary}
If in \textnormal{Theorem~\ref{thm:GMMH* UNI}} we let $n$ be a prime then we obtain \textnormal{Theorem~\ref{thm:MMH* UNI}}.
\end{corollary}

\begin{proof}
When $n$ is prime, $\gcd_{\mathbf{a}=\langle a_1,\dots,a_k \rangle \in \mathbb{Z}_n^k \setminus \lbrace \mathbf{0} \rbrace}(a_1, \ldots, a_k, n)=1$, so we get $\triangle$-universality.
\end{proof}

We remark that if in the family GMMH$^*$ we let the keys $\mathbf{x}=\langle x_1, \ldots, x_k \rangle \in \mathbb{Z}_n^k$ satisfy the general conditions $\gcd(x_i,n)=t_i$ ($1\leq i\leq k$), where $t_1,\ldots,t_k$ are given positive divisors of $n$, then the resulting family, which was called GRDH in \cite{BKSTT2}, is no longer `always' $\varepsilon$-A$\triangle$U. In fact, it was shown in \cite{BKSTT2} that the family GRDH is $\varepsilon$-A$\triangle$U for some $\varepsilon<1$ if and only if $n$ is odd and 
$\gcd(x_i,n)=t_i=1$ (that is, $x_i \in \mathbb{Z}_n^{*}$) for all $i$. Furthermore, if these conditions are satisfied then GRDH is $\frac{1}{p-1}$-A$\triangle$U, where $p$ is the smallest prime divisor of $n$ (this bound is also tight). This result is then applied in giving a generalization of a recent authentication code with secrecy. A key ingredient in the proofs in \cite{BKSTT2} is an explicit formula for the number of solutions of restricted linear congruences (a restricted version of Proposition~\ref{Prop: lin cong}), recently obtained by Bibak et al. \cite{BKSTT}, using properties of Ramanujan sums and of the finite Fourier transform of arithmetic functions.

\section*{Acknowledgements}

The authors would like to thank the editor and anonymous referees for helpful comments. During the preparation of this work the first author was supported by a Fellowship from the University of Victoria (UVic Fellowship).

\end{document}